\documentclass[10pt,conference]{IEEEtran}

\usepackage[utf8]{inputenc}
\usepackage{amsmath,amssymb,amsthm,enumerate}
\usepackage{mathtools}
\usepackage{graphicx}
\usepackage{fancyhdr}
\usepackage{mathbbol}
\usepackage{listings}
\usepackage{subfig}
\usepackage{placeins}
\usepackage{bm}
\usepackage{bbm}
\usepackage{upgreek}
\usepackage{array}
\usepackage{comment}
\usepackage{makecell}
\usepackage[backend=biber,style=ieee, maxnames=10]{biblatex}
\bibliography{IEEEabrv,mybibfile}

\usepackage{cuted}
\usepackage{lipsum}

\graphicspath{ {./images/} }

\newtheorem{theorem}{Theorem}
\newtheorem{lemma}{Lemma}
\newtheorem{proposition}{Proposition} 
\newtheorem{definition}{Definition}
\newtheorem{corollary}{Corollary}

\newcommand{\RomanNumeralCaps}[1]
    {\MakeUppercase{\romannumeral #1}}

\newcommand\numeq[1]%
  {\stackrel{\scriptscriptstyle(\mkern-1.5mu#1\mkern-1.5mu)}{=}}
\newcommand{\argmin}{\operatornamewithlimits{argmin}}

\usepackage{stackengine}
\def\delequal{\mathrel{\ensurestackMath{\stackon[1pt]{=}{\scriptstyle\Delta}}}}

\DeclareMathOperator*{\argmax}{arg\,max}

\begin{document} 
\title{Analysis of QoS in Heterogeneous Networks with Clustered Deployment and Caching Aware Capacity Allocation}

\author{\IEEEauthorblockN{Takehiro Ohashi\IEEEauthorrefmark{1}}
\IEEEauthorblockA{Department of Mathematical and Computing Science,
Tokyo Institute of Technology\\
Email: \IEEEauthorrefmark{1}ohashi.t.ag@m.titech.ac.jp }}

\maketitle

\begin{abstract}
In cellular networks, the densification of connected devices and base stations engender the ever-growing traffic intensity, and caching popular contents with a smart management is a promising way to alleviate such consequences. Our research extends the previously proposed analysis of three-tier cache enabled Heterogeneous Networks (HetNets). The main contributions are threefold. We consider the more realistic assumption; that is, the distribution of small base stations is following Poisson-Poisson cluster processes, which reflects the real situations of geographic restriction, user dense areas, and coverage-holes. We propose the allocation of downlink data transmission capacity according to the cases of requested contents which are either cached or non-cached in nearby nodes and elucidate the traffic efficiency of the allocation under the effect of clustered deployment of small base stations. The throughput and delay of the allocation system are derived based on the approximated sojourn time of Discriminatory Processor Sharing (DPS) queue. We present the results of achievable efficiency and such a system's performance for a better caching solution to the challenges of future cellular networks.

\vspace{5mm}
\end{abstract}

\section{Introduction}
The ever-growing communication traffic and capacity demand in cellular networks with the rise of complex networks composed of heterogeneity of cells necessitate smart traffic management. To expand the fifth-generation cellular systems (5G) network or even proceed beyond sixth-generation cellular systems (6G) in reality, the challenge is to support hundreds of gigabits of traffic from user's devices to the core network through the backhauls as well as the high throughput of user-to-node data transmission while satisfying the extreme requirements such as availability, latency, energy, and cost-efficiency \cite{Giordani2020, Jaber2016,Parvez2018}. To this end, caching popular contents in the networks was proposed and analysed extensively for it's potentials \cite{Giordani2020, Jaber2016,Parvez2018,Bastug2015,Bastug2016}. In this paper, we extend the previous analysis of the cache-enabled heterogeneous network conducted by \cite{Yang2016}. The three-tier network, consisting of base stations (BSs), relays, and device-to-device sharing links (D2Ds) was considered, and their result showed that the global throughput of the proposed caching system can be increased significantly. To carry out such a system, in reality, more analysis is needed because of the complexity of the spatial dynamics of cellular networks. We also tackle the challenge of allocating network resources for better efficiency under this complexity. 

The contributions of this work are three-fold. To reinforce the result of the previous analysis, instead of considering all three tiers following independent Poisson point process (PPP), we consider that the deployment of small cells (either relay or small base stations) is distributed as Thomas cluster process (TCP), a class of Poisson-Poisson cluster processes (PPCP), which exhibit attractive point patterns and reflect the real situations of small cell deployment influenced by geographic restrictions, user dense areas and coverage-holes. This assumption of the attractive pattern of small cell deployment is in-negligible when we analyze the performance of the caching system since the performance of downlink data transmission is largely influenced by whether the requested contents are cached nearby cells or not. Next, we propose a new system for more efficient caching solutions, that is, allocating the downlink data transmission capacity according to the different circumstances of contents requests from users. This capacity allocation system is modeled by DPS queue, a variant of $\textit{multi-class processor sharing queue}$, which assigns various weights on each service rate of different user classes and enables the service differentiation. The previous work in \cite{Yang2016} used egalitarian processor sharing (EPS), a simple queuing model that shares service rate equally. The property of flow-level performance of data transmission in a cellular network was well studied using PS queue as well as DPS queue \cite{yujin2007,Mohamed2013} but without taking an account of complex spatial stochastic dynamics of users and base stations. The study of cellular networks is very limited without considering spatial stochastic dynamics since the signal to interference noise ratio (SINR) perceived by a user depends on complicated interaction among the spatial distribution of base stations. Recently, successful analyses are using both stochastic geometry approach and queuing theory to overcome this complex spatial dynamics such as \cite{Yang2016,Blaszczyszyn2015, Blaszczyszyn2019}. Thus far, EPS queue and generalized processor sharing (GPS) queue were used while the DPS queue was not the case. Mathematically, DPS queue is complex and some tractability is lost, and many analyses are limited under special circumstances \cite{Fayolle1980,Avrachenkov2005, Altman2006}. Instead of using the exact analysis of the DPS queue, we use the best approximation derived by \cite{izagirre2017} to show the possible improvement of a cellular network traffic efficiency by the allocation system. To the best of our knowledge, this is the first application DPS queue while taking an account of the effect of spatial stochastic patterns of users and base stations to analyze the capacity allocation performance of HetNets. Lastly, we compare numerically the results obtained by \cite{Yang2016} with clustered small cells and our capacity allocated system. We show that the clustered deployment of small cells differs the traffic of entire networks from the uniformly distributed one and by prioritizing downlink data transmission of a particular class of users the throughput of entire networks can be improved further. 

\subsection{Paper Organization}
The rest of the paper is organized as follows. In section \RomanNumeralCaps{2}, we introduce the system model of our three-tier HetNets and protocols of caching as well as tier association probability among users in different circumstances. Note that the tier association probabilities derived in this section are needed to express the average downlink data transmission rate of a typical user as well as the probability of a user active in different circumstances. In section \RomanNumeralCaps{3}, we derive the average downlink data transmission rate called average ergodic rate. In section \RomanNumeralCaps{4}, we derive QoS metrics such as the mean number of requests, delay, and throughput under the caching aware capacity allocation system using approximated sojourn time of DPS queue. The numerical results and the performance of such a system are shown in section \RomanNumeralCaps{5}. Ultimately, we give our conclusions and discussion in section \RomanNumeralCaps{6}. 

\section{System Model}
We introduce our three-tier HetNets consisting of macro base stations (MBSs), small base stations (SBSs), and device-to-device sharing links (D2Ds). Note that our work follows the idea from \cite{Yang2016} and regards their result as a baseline; here we consider SBSs in our system model that follows TCP while their analysis considers relays that follow independent PPP. The tier association probabilities in this network are also derived to characterize the request of contents from users to nearby base stations (BSs) or D2Ds. The list of notations is provided in TABLE \RomanNumeralCaps{1} for the reader to review some of the notations used hereafter.

\begin{table}[ht]
\centering
\caption{LIST OF NOTATIONS.}
\begin{tabular}{| c | c |}
\hline
$\Phi_i$& The point process of the i-th BS/D2D tier  \\
\hline
$f_{cd_i}, \;\;\; \Bar{F}_{cd_i}$ & PDF and CCDF of contact distance of $\Phi_i$ \\
\hline
$f_d(s)$ & The displacement kernel defined $\exp(\frac{-\|s\|^2}{2\sigma^2})/2\pi\sigma^2$ \\
\hline
$\tau_i(r)$ & $\begin{cases}
      \bar{m} \int^{\infty}_0 2\pi \lambda_{p_2} \frac{z}{\sigma} q(\frac{z}{\sigma},\frac{r}{\sigma}) \times \\ \exp(-\bar{m}(1-Q_1(\frac{z}{\sigma},\frac{\frac{P_2}{P_i}^{\frac{1}{\beta}}r}{\sigma})))dz ,& \text{if } i = 2 \\
    2 \pi \lambda_i r                      ,& \text{if } i = 1, 3 
\end{cases}$ \\
\hline
$f_{\mathcal{S}_i}, f_{\Hat{\mathcal{S}}_i}, f_{\mathcal{S}_{1,j}}$ & \makecell{PDF of distance of a serving node given an event \\ $\mathcal{S}_i,\mathcal{S}_i,\mathcal{S}_{1,j}$} \\
\hline
$\Bar{F}_{\mathcal{S}_i}, \Bar{F}_{\Hat{\mathcal{S}}_i}, \Bar{F}_{\mathcal{S}_{1,j}}$ & \makecell{CCDF of distance of a serving node given an event \\ $\mathcal{S}_i,\mathcal{S}_i,\mathcal{S}_{1,j}$} \\
\hline
${}_{2}F_1[a,b;c;d]$ & Gauss hyper-geometric function \\
\hline
$\alpha$ & Ratio of cache-enabled users ($\alpha \in [0,1]$) \\
\hline
$\beta$ & Path loss exponent ($\beta=4$) \\
\hline

\end{tabular}
\end{table}%

\subsection{Network Architecture}
We consider three-tier HetNet consisting of MBSs, SBSs, D2Ds. All SBSs are connected to the nearest MBS through backhauls, and then MBSs are connected to a core network (see Fig. 1). We define that the nodes of users and MBSs follow independent homogenious PPPs in $\mathbb{R}^2$ denoted as $\Phi_i$ with intensity $\lambda_i$ for $i=0,3$. The nodes of SBSs follow PPCP denoted as $\Phi_2$ in $\mathbb{R}^2$ with conditional intensity function given the parental point process $\lambda_2(\Vec{x}) = \bar{m} \sum_{\Vec{z} \in \Phi_{p_2}} f_d(\Vec{x}|\Vec{z})$ where $\Phi_{p_2}$ is the parental point process following PPP with intensity $\lambda_{p_2}$ and $\bar{m}$ is the mean number of daughter points distributed according to some distribution of $f_d(\Vec{x}|\Vec{z})$ around the center of each parental point $\Vec{z}$. In this paper for the ease of our analysis, we assume that the daughter points are independently and normally distributed around a parental point $\Vec{z}$, i.e. $f_d(x)= \exp(-\|x\|^2/(2\sigma^2))/(2\pi \sigma^2)$ which is called the Thomas cluster point process (TCP). The formal definition of PPCP is

\begin{definition}
Stationary PPCP on $\mathbb{R}^2$
\begin{equation}
    \Phi_{2}= \cup_{X_i \in \Phi_{p_2}}(X_i + \Psi_i),
\end{equation}
where $\Phi_{p_2}=\{X_i\}_{i=1}^{\infty}$ is the parental point process, homogeneous PPP with intensity $\lambda_{p_2}$ and a mark $\Psi_i = \{Y_{i,j}\}_{j=1}^{N_i}$ of the point $X_i$ is in-homogeneous PPP with the intensity function $\lambda_d(\Vec{y})$ s.t. $ \int_{\mathbb{R}^2} \lambda_d(\Vec{y})d\Vec{y} = \int_{\mathbb{R}^2} \Bar{m}f_d(\Vec{y})d\Vec{y} = \Bar{m}$ is the mean number of daughter points. $f_d$ is the density function of each daughter point around its parent.
\end{definition}

When a PPCP is conditioned on the parental point process, it is an in-homogeneous PPP. The more extensive description of PPCP can be found in \cite{Miyoshi19,Saha2018}.
Let $\mathcal{B}(\mathbb{R}^2)$ be the Borel $\sigma$-algebra on $\mathbb{R}^2$ and $A \subset \mathcal{B}(\mathbb{R}^2)$. Conditioned on $\Phi_{p_2}$, the counting measure of $\Phi_2$ is random measure which counts the number of points of $\Phi_2$ falling in A and is given by ( see \cite{Miyoshi19,Saha2018})
\begin{equation}
    N_{\Phi_2}(A)|\Phi_{p_2} \thicksim Poisson(\Bar{m} \sum_{\Vec{x} \in \Phi_{p_2}} \int_A f_d(\Vec{y}-\Vec{x})d\Vec{y}).
\end{equation}
We consider \cite{Yang2016} as a baseline in this paper, and in order to compare the results we let the intensity of SBSs in baseline case (which is assumed to follow independent homogeneous PPP) be the same as our scenario ($\Bar{m} \lambda_{p_2} = \lambda_2$). \par
For the user tier, we have the proportion of cache-enabled users who can transmit some contents cached in their local storage to some users who are not cache-enabled. Those cache enabled users are distributed as a thinning of the homogeneous PPP of $\Phi_1$ with intensity $\lambda_1 = \alpha \lambda_0$, where $\alpha \in [0,1]$ and $\lambda_0$ is the intensity of the nodes of users. Considering the real circumstances, we set the intensity as $\lambda_0 \gg \lambda_2 > \lambda_3$. \par
There ia $N$ total number of contents, and all of them are stored in MBSs. We assume all the contents are the same size denoted as $S \text{ [bits]}$. All SBSs can cache $M_2$ number of contents and have a caching storage with the size $M_2 \times S \text{ [bits]}$. Similarly, the cache enabled users can cache $M_1$ number of contents and have a caching storage with the fixed size $M_1 \times S \text{ [bits]}$. It is a natural to consider $M_1 \ll M_2 \ll N$. See Fig. 1, the blue disks are represented as caching of $M_1$ number of contents which all the cache enabled users to have and contain the same copies of the contents. SBSs have $M_2$ number of contents which includes $M_1$ number of contents as well. MBSs have all of the contents stored in their storage. The i-th popular content requested by the typical user follows Zipf distribution, and for arbitrary content popularity ranks $a,b$ with $a<b$:
\begin{equation}
    f_i = \frac{1/i^\gamma}{\sum_{j=1}^N \frac{1}{j^\gamma}},\: \: \: \: \: \: \sum_{i=a}^{b} f_i \delequal F_{pop}(a,b),
\end{equation}
where $\gamma \geq 0$ is the parameter of skewness of the content popularity distribution. Then, we assume all the cache-enabled users have cached $1$st to $M_1$-th popular contents in their storage. This implies that the probability of a content requested by the typical user that is stored in cache-enabled user is $F_{pop}(1,M_1)$. Similarly, for all SBSs, $1$st to $M_2$-th popular contents are cached in the storage. Therefore, the probability of a content requested by the typical user that is stored in SBSs is $F_{pop}(1,M_2)$.

\begin{figure}[!htbp]
\includegraphics[scale=0.25]{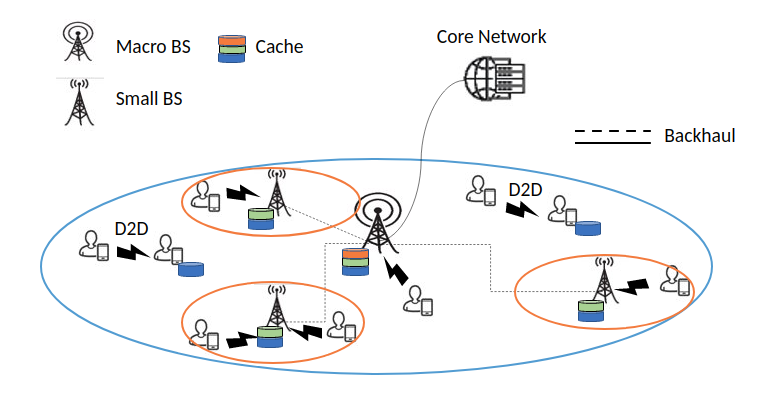}
\centering
\caption{Illustration of three tier cache enabled heterogeneous networks. The blue circle is the coverage area of MBS. The orange circle is the coverage area of SBS.}
\end{figure}

\subsection{Access and Cache Protocol}
For the completeness of this paper, we summarize all the assumption from \cite{Yang2016} and emphasis some of the important concepts here. We assume all users are following the max-power association rule. This means that all users request contents from a node which provides the highest power that is the closet node in a tier.  
For $i=1,2,3$, the maximum received power from the $i$-th tier is
\begin{equation}
    C_{i} = \nu B_{i} P_{i} \|\Vec{x}_{i}^*\|^{-\beta},
\end{equation}
where $\|\Vec{x}_{i}^*\|$ is the distance from the typical user at $(0,0)$ to its closest BS in $i$-th tier, $\Vec{x}_{i}^* \in \Phi_i$; that is $\Vec{x}^*_i = \argmax_{\Vec{x} \in \Phi_i} P_i \|\Vec{x}\|^{-\beta} = \argmin_{\Vec{x} \in \Phi_i}\|\Vec{x}\|$. $P_i$ is a transmit power of the node in $i$-th tier, and $\nu$, $B_{i}$, and $\beta$ are the propagation constant, bias, and path-loss exponent, respectively. For the ease of analysis, we let $\nu = 1$ and $B_{i} = 1$. In this paper we consider the path-loss exponent $\beta=4$.
Under this max-power association rule there are four different cases that the typical user connects to corresponding BSs/D2Ds.
\begin{itemize}
         \item Case 1: when the typical user is not cache-enabled, the user requests contents to a nearby node, either MBS, SBS, or cache-enabled user (D2D).
        \item Case 2: when the typical user is cache-enabled and the requested contents by the user are not cached in the storage, the only choice is to request the contents to either MBS or SBS.
        \item Case 3: when the typical user is not cache-enabled and a cache-enabled user (D2D) provides the highest power but the requested contents by the user were not cached in the storage of the cache-enabled user, the user requests the contents to either MBS or SBS.
        \item Case 4: when the typical user is cache-enabled and the requested contents by the user were cached in his/her local storage, the contents were retrieved and reused immediately.
\end{itemize}
 Fig. 2. is the illustration of all four cases. We assume that cache is placed according to proactive-caching: all contents are distributed to the storage of BSs/D2Ds during off-peak and ready to use. Additionally, when the typical user has requested contents to an SBS and the contents are not stored in the storage of the SBS, it requires wired backhauls to fetch requested contents from the storage of the nearest MBS. This is the case when the downlink data transmission is $\textit{backhaul-needed}$ (BH-needed), and we assume the typical user under BH-needed transmission experiences some delay.
 
\begin{figure}
    \centering
    \includegraphics[width=0.14\textwidth]{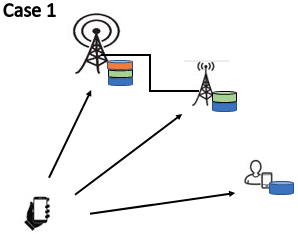}
    \includegraphics[width=0.14\textwidth]{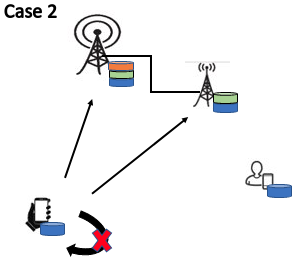}
    \includegraphics[width=0.14\textwidth]{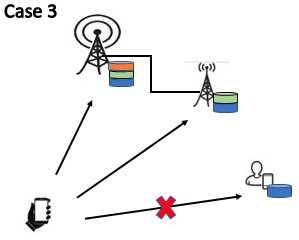}
    \includegraphics[width=0.14\textwidth]{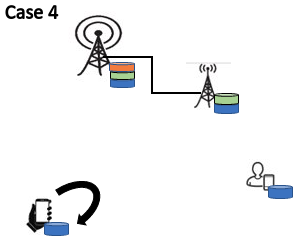}
    \caption{Illustration of contents requests in different cases.}
    \label{fig:foobar}
\end{figure}

\subsection{Tier Association Probability}

We derive the probability of ordering of maximum power received by the typical user from each tier $\Phi_i$ for $i=1,2,3$, where $\Phi_1$ and $\Phi_3$ follow independent PPP and $\Phi_2$ follows independent TCP. This probability is used to characterize the requests of contents from the typical user to BSs/D2Ds under the max-power association rule explained in the previous subsection. It is commonly called the tier association probability which is the probability that the typical user communicates to a tier. \par
The PDF of contact distance of $\Phi_i$ for $i \in \{1,3\}$ denoted as $f_{cd_i}$ and the complementary CDF denoted as $\Bar{F}_{cd_i}$ are
\begin{equation}
    f_{cd_i}(r_i) = 2\pi \lambda_i r_i e^{-\pi \lambda_i r_i^2}, \; \; \; \Bar{F}_{cd_i}(r_i) = e^{-\pi \lambda_i r_i^2}.
\end{equation}
 The conditional PDF of contact distance of $\Phi_2$ given the parental point process $\Phi_{p_2}$ is derived (see, Lemma 1 \cite{Miyoshi19,Saha2018}) and given by 
\begin{multline}
    f_{cd_2}(r_2|\Phi_{p_2}) = \Bar{m} \sum^{\infty}_{i=1} \frac{1}{\sigma} q(\frac{\|X_i\|}{\sigma},\frac{r_2}{\sigma}) \times  \\ \prod_{i=1}^\infty \exp(-\Bar{m}(1-Q_1(\frac{\|X_i\|}{\sigma},\frac{r_2}{\sigma}))),
\end{multline}
where $Q_1$ is the first-order Marcum Q-function  $Q_1(a,b) = \int^\infty_{b} z \exp(-\frac{z^2+a^2}{2})I_0(az)dz$ and $I_0(z)=\pi^{-1} \int^\pi_0 e^{z\cos\phi}d\phi$ is the modified Bessel function of the first kind with order zero. $q(a,b)$ is PDF of Rician distribution, and its cumulative distribution function is $1-Q_1(a,b)$. Note that the CDF of contact distance distribution of TCP was derived in \cite{Afshang}, and the complementary CDF is given by
\begin{multline}
    \Bar{F}_{cd_2}(r_2) = \exp(-\int^{\infty}_0 2\pi \lambda_{p_2} z (1-\exp(-\bar{m}\\(1-Q_1(\frac{z}{\sigma},\frac{r_2}{\sigma}))))dz).
\end{multline}
Then, by using (6) we can derive the PDF of contanct distance distribution of TCP in the following Lemma 1.
\begin{lemma}
The PDF of contact distance distribution of TCP is
\begin{multline}
    f_{cd_2}(r_2) =\mathbb{E}_{\Phi_{p_2}}[f_{cd_2}(r_2|\Phi_{p_2})]  \\ 
    = \mathbb{E}_{\Phi_{p_2}}[ \Bar{m} \sum^{\infty}_{i=1} \frac{1}{\sigma} q(\frac{\|X_i\|}{\sigma},\frac{r_2}{\sigma}) \prod_{i=1}^\infty \exp(-\Bar{m}(1-Q_1(\frac{\|X_i\|}{\sigma},\frac{r_2}{\sigma})))] \\
    = \int^{\infty}_0 \bar{m} \lambda_{p_2} 2\pi \frac{z}{\sigma} q(\frac{z}{\sigma},\frac{r_2}{\sigma}) \exp(-\bar{m}(1-Q_1(\frac{z}{\sigma},\frac{r_2}{\sigma})))dz \times  \\ 
    \exp(-\int^{\infty}_0 2\pi \lambda_{p_2} z (1-\exp(-\bar{m}(1-Q_1(\frac{z}{\sigma},\frac{r_2}{\sigma}))))dz).
\end{multline}
\begin{proof}
 We can directly apply the sum-product functional with respect to parental PP which follows PPP (see Lemma 6 \cite{Saha2018}).
\end{proof}
\end{lemma}
Next, we derive the probability of ordering the maximum power received by the typical user from each tier.
\begin{proposition}
The ordered tier association probability of three tier $i,j,k \in \{1,2,3\}$, $i \neq j \neq k$ is
\begin{multline}
    \mathbb{P}(C_i > C_j > C_k) =  \\
    \int^{\infty}_{0} \int^{\infty}_{(\frac{P_{j}}{P_{i}})^{\frac{1}{\beta}}r_{i}} \int^{\infty}_{(\frac{P_{k}}{P_{j}})^{\frac{1}{\beta}}r_{j}}
    \int^{\infty}_0 \bar{m} \lambda_{p_2} 2\pi \frac{z}{\sigma} q(\frac{z}{\sigma},\frac{r_2}{\sigma}) \times \\ \exp(-\bar{m}(1-Q_1(\frac{z}{\sigma},\frac{r_2}{\sigma})))dz \times \\ \exp(-\int^{\infty}_0 2\pi \lambda_{p_2} z (1-\exp(-\bar{m}(1-Q_1(\frac{z}{\sigma},\frac{r_2}{\sigma}))))dz) \times \\
    \prod_{l \in \{1,3\}}2\pi \lambda_l r_l \exp(-\pi\lambda_{l} r_{l}^2) dr_{k} dr_{j} dr_{i}.
\end{multline}
\end{proposition}
\begin{proof}
Let $R_i, R_j$ and $R_k$ be the random variable of contact distance to tier $i,j$ and $k$, respectively. Then, we have the joint PDF of contact distance scaled by power from each tier as
\begin{multline}
    \mathbb{P}(C_i > C_j > C_k) = \mathbb{P}((\frac{P_i}{P_k})^{-\frac{1}{\beta}}R_k > (\frac{P_i}{P_j})^{-\frac{1}{\beta}}R_j > R_i) \\
    = \int^{\infty}_{0} \int^{\infty}_{(\frac{P_{j}}{P_{i}})^{\frac{1}{\beta}}r_{i}} \int^{\infty}_{(\frac{P_{k}}{P_{j}})^{\frac{1}{\beta}}r_{j}}
    \int^{\infty}_0 f_{cd_k}(r_k)f_{cd_j}(r_j) \\f_{cd_i}(r_i)dr_k dr_j dr_i.
\end{multline}
Since the PDF of contact distance distribution of each tier is given by (5) and (8), plugging in them into (10) gives the desired equation.
\end{proof}
The association probability conditional on parental PP is derived in lemma 3 \cite{Saha2018}. They assume an arbitrary many number of tiers in HetNets.  Similarly, We can derive the unconditional tier association probability by the following.
\begin{corollary}
The unconditional $i$-th tier association probability is
\begin{multline}
    \mathbb{P}(C_i > \max_{\forall n \neq i} C_n) = \int^{\infty}_0 \tau_i (r) \exp(-\int^{\infty}_{0} 2 \pi z \lambda_{p_2} \times \\
    (1-\exp(-\bar{m}(1-Q_1(\frac{z}{\sigma},\frac{(\frac{P_2}{P_i})^{\frac{1}{\beta}}r}{\sigma}))))dz) \times \\
    \exp(-\sum_{l \in \{1,3\}} \pi \lambda_l (\frac{P_l}{P_i})^{\frac{2}{\beta}} r^2)dr,
\end{multline}
where $\tau_i(r)$ is defined in TABLE \RomanNumeralCaps{1}, and $n \in \{1,2,3\}$.
\begin{proof}
In the case $i=2$,
\begin{multline}
    \mathbb{P}(C_i> \max_{\forall n \neq i} C_n) = \mathbb{P}(\bigcap_{j\in \{1,3\}}P_2 R_2^{-\beta} > P_j R_j^{-\beta}) \\
    = \mathbb{P}(\bigcap_{j\in \{1,3\}} R_j > (\frac{P_j}{P_2})^{\frac{1}{\beta}}R_2) \\
    = \int^{\infty}_0 \prod_{j\in \{1,3\}} \Bar{F}_{cd_j}((\frac{P_j}{P_2})^{\frac{1}{\beta}}r) f_{cd_2}(r)dr.
\end{multline}
In the case $i=1,3$, let $j=\{1,3\} \setminus \{i\}$, then we have
\begin{multline}
    \mathbb{P}(C_i> \max_{\forall n \neq i} C_n) = \mathbb{P}(R_2>(\frac{P_2}{P_i})^{\frac{1}{\beta}}R_i) \mathbb{P}(R_j>(\frac{P_j}{P_i})^{\frac{1}{\beta}}R_i)\\
    =\int^{\infty}_0 \prod_{j\in \{1,3\}} \Bar{F}_{cd_2}((\frac{P_2}{P_i})^{\frac{1}{\beta}}r) \Bar{F}_{cd_j}((\frac{P_j}{P_i})^{\frac{1}{\beta}}r) f_{cd_i}(r)dr.
\end{multline}
Plugging (5), (7) and (8) into (12) and (13); then, rewriting the plugged equation by $\tau_i(r)$ gives our result in (11).
\end{proof}
\end{corollary}
We define the notations and use them for the convenience hereafter. $ \mathbb{P}(C_i> \max_{\forall n \neq i} C_n) \overset{\Delta}{=}  \mathcal{G}_{3,i}$, and also, we denote $\mathbb{P}(C_i>C_j>C_k)$ and $\mathbb{P}(C_i>C_j)$ as $\mathbb{P}_{i,j,k}$ and $\mathbb{P}_{i,j}$, respectively. Note that $\mathbb{P}(C_i>C_j)$ is two tier case of proposition 1, and the derivation is the same and omitted.

\section{Average Ergodic Rate}
In this section, the average ergodic rate of the downlink is analyzed under the three-tier HetNets described in the previous section. We assume that the communication links between BSs/D2Ds to users share the same frequency bandwidth, and this yields the interference among the links. Recall that there are four cases of the typical user's requests of contents, and the rate of average downlink data transmission differs in those different cases (see Fig. 2). 

\subsection{Active BSs and D2D Links}
We define all active nodes which yield the interference. We denote $\lambda'_i$ by the intensity of active nodes in tier $i$ for $i=1,2,3$. For the tier 1, not all cache-enabled users need to be active when the number of requests from non-cache-enabled users through the D2D transmission is less than the cache-enabled users. Thus, we define the intensity of active D2D nodes as $\lambda_1^{'}  \overset{\Delta}{=} \min \{\alpha \lambda_0, \lambda_0  \mathcal{G}_{3,1} (1-\alpha)  F_{pop}(1,M_1)\}$. For the other tiers, since all MBSs and SBSs are active over the time, we define $\lambda'_i \overset{\Delta}{=} \lambda_i$. We also define $\Phi_1'$ as the thinned PPP of active D2D links and $\Phi_2' \overset{\Delta}{=} \Phi_2$ and $\Phi_3' \overset{\Delta}{=} \Phi_3$, where active MBSs and SBSs are unchanged from original MBSs and SBSs.

\subsection{Signal to Interference Plus Noise Ratio}
The downlink power received from the serving BS/D2D at $\Vec{x}_i \in \Phi_i'$ is $P_i h_{\Vec{x}_i} \|\Vec{x}_i\|^{-\beta}$ for $i=1,2,3$, and the signal-to-interference-plus-noise ratio (SINR) of the typical user located at the origin connecting to that node is
\begin{equation}
    SINR_i(\|\Vec{x}_i\|) = \frac{P_i h_{\Vec{x}_i} \|\Vec{x}_i\|^{-\beta}}{\sum^{3}_{j=1} \sum_{\Vec{y} \in \Phi_j' \setminus \{\Vec{x}_i\}}P_j h_{\Vec{y}} \|\Vec{y}\|^{-\beta}+ N_0},
\end{equation}
where $h_{\Vec{x_i}}$ and $h_{\Vec{y}}$ follow Rayleigh fading distributed as exponential with unit mean. In the denominator, all 
 $\|\Vec{y}\|$ is the distance between reference user to its interfering active nodes in $j$-th tier. $N_0$ is thermal noise.
Let $x$ be the distance of a requesting user to the serving node. The average ergodic rate of the typical user when it communicates with the i-th tier is given by \cite{Yang2016}:
\begin{equation}
    \mathcal{U}_i \delequal \mathbb{E}_{x} [\mathbb{E}_{SINR_{i}}[\ln(1+SINR_{i}(x))|\|\Vec{x}_i\|=x]],
\end{equation}
where the expectation of $SINR_i$ is with respect to all $\Phi_j'$ for $j=1,2,3$ and the fading effect $h_{\Vec{x_i}}$ and $h_{\Vec{y}}$, and then the expectation of $\|\Vec{x}_i\|$ is over the distance between the typical user and its serving node. The equation above indicates the average ergodic rate, which is the mean rate of data transmitted, of a randomly chosen user associated with the $i$-th tier in a cell.

\subsection{The Average Ergodic Rate in Case 1}
Let us define the association event such that the typical user requests contents to the $i$-th tier as $\mathcal{S}_i = \{P_i \|\Vec{x}^*_i\|^{-\beta} > \max_{n \neq i} P_n \|\Vec{x}^*_{n}\|^{-\beta} : n \in \{1,2,3\} \}$. The PDF of the distance distribution of the serving node conditional on parental PP and $\mathcal{S}_i$ is given in Lemma 4, \cite{Saha2018}. Then, we can obtain the PDF of the distance between the typical user and the associated serving node unconditional on $\Phi_{p_2}$ and conditional on an event $\mathcal{S}_i$.
\begin{lemma}
The PDF of the distance from the typical user associated to a node in the $i$-th tier for $i=1,2,3$ is
\begin{multline}
    f_{\mathcal{S}_i}(x) = \frac{\tau_i(x)}{\mathcal{G}_{3,i}} \exp(-\int^{\infty}_0 2\pi \lambda_{p_2} z \times \\ (1- \exp(-\bar{m}(1-Q_1(\frac{z}{\sigma},\frac{(\frac{P_2}{P_i})^{\frac{1}{\beta}}x}{\sigma}))))dz) \times\\
    \exp(- \pi \sum_{j \in \{1,3\}} \lambda_{j} (\frac{P_j}{P_i})^{\frac{2}{\beta}} x^2),
\end{multline}
where $\tau_i(x)$ is given in the TABLE 1. 
\begin{proof}
The CCDF of $f_{\mathcal{S}_i}(x)$ can be expressed as
\begin{align}
    \Bar{F}_{\mathcal{S}_i}(r| \mathcal{S}_i) &= \mathbb{P}(\|\Vec{x}_{i}^*\|>r| \mathcal{S}_i)\\
    &= \frac{\mathbb{P}(R_i>r, \mathcal{S}_i)}{\mathbb{P}(\mathcal{S}_i)} \\
    &= \frac{1}{\mathcal{G}_{3,i}} \mathbb{P}(\bigcap_{j\in \{1,2,3\}\setminus \{i\}} P_i R_i^{-\beta}> P_j R_j^{-\beta}, R_i > r).
\end{align}
Then, the last equation is similar to what we derived in Corollary 1, and by taking the derivative with respect to $r$ we get the desired result.
\end{proof}
\end{lemma}
Since we know that the distribution of the distance of the typical user and the associated serving node, we are ready to obtain the average ergodic rate in case 1 as follows.
\begin{theorem}
The average ergodic rate of the typical user connecting to a node in the $i$-th tier in case 1 is
\begin{multline}
    \mathcal{U}_{1,i} = \int^{\infty}_0 \int^{\infty}_0 \frac{\mathcal{M}_y(s,x)}{1+s}ds \frac{\tau_i(x)}{\mathcal{G}_{3,i}} \times \\ 
   \exp(-\int^{\infty}_0 2\pi \lambda_{p_2} z (1-\exp(-\bar{m}(1-Q_1(\frac{z}{\sigma},\frac{(\frac{P_2}{P_i})^{\frac{1}{\beta}}x}{\sigma}))))dz) \times \\
    \exp(- \pi \sum_{j \in \{1,3\}} \lambda_{j} (\frac{P_j}{P_i})^{\frac{2}{\beta}} x^2) dx,
\end{multline}
where 
\begin{align*}
&\mathcal{M}_y(s,x) = \\
&\prod_{j\in \{1,3\}} \exp(-2\pi \lambda'_j x^2 \frac{s (\frac{P_j}{P_i})^{\frac{2}{\beta}}}{\beta(1-\frac{2}{\beta})} {}_{2}F_1[1,1-\frac{2}{\beta};2-\frac{2}{\beta};-s])\\
&\exp(-2 \pi \lambda_{p_2} \int_{0}^{\infty}  (1-\exp(-\Bar{m} \int_{(\frac{P_2}{P_i})^{\frac{1}{\beta}}x}^{\infty} \frac{1}{1+(s\frac{P_2}{P_i})^{-1}(\frac{y+z}{x})^{\beta}} \\ 
& \frac{1}{\sqrt{2\pi}\sigma^2} \exp(-\frac{y^2}{2\sigma^2}) dy))zdz).
\end{align*}
and ${}_{2}F_1[a,b;c;d]$ is Gauss hyper-geometric function.
\end{theorem}
\begin{proof}
See Appendix A.
\end{proof}

\subsection{The Average Ergodic Rate in Case 2}
In case 2, the typical user requests contents to either MBSs or SBSs. Therefore, by following the similar derivation as Lemma 1, given an event $\Hat{\mathcal{S}}_i = \{ P_i \|\Vec{x}^*_i\|^{-\beta} > \max_{j \neq i} P_j \|\Vec{x}^*_{j}\|^{-\beta}: j \in \{2,3\}\}$ the unconditional PDF of the distance between the typical user and the serving node in the $i$-th is
\begin{multline}
    f_{\Hat{\mathcal{S}}_i}(x) = \frac{\tau_i(x)}{\mathbb{P}_{i,j}} \exp(-\int^{\infty}_0 2\pi \lambda_{p_2} z \times \\ (1-\exp(-\bar{m}(1-Q_1(\frac{z}{\sigma},\frac{(\frac{P_2}{P_i})^{\frac{1}{\beta}}x}{\sigma}))))dz) \times\\
    \exp(- \pi \lambda_{3} (\frac{P_3}{P_i})^{\frac{2}{\beta}} x^2),
\end{multline}
where $i,j \in \{2,3\}$ s.t. $i\neq j$ and $\tau_i(x)$ is defined in TABLE 1. The derivation is similar to Lemma 2 and therefore omitted. For the average ergodic rate of case 2, the interfering active D2D nodes can be closer than a serving node, which is either MBS or SBS. This is because the typical user is cache-enabled and can not request the contents from the nearest cache-enabled user through the D2D link (see Fig. 2). Therefore, the possible distance between a requested user to the closest cache-enabled user denoted as $a$ can be as close as 0; that is, the possible distance $a \rightarrow 0$. Then, by considering this we have the following theorem.
\begin{theorem}
The average ergodic rate of the typical user connecting to a node in the $i$-th tier in case 2 is
\begin{multline}
    \mathcal{U}_{2,i} = \int^{\infty}_0 \int^{\infty}_0 \frac{\mathcal{M}_y(s,x)}{1+s}ds 
    \frac{\tau_i(x)}{\mathbb{P}_{i,j}} \times \\
    \exp(-\int^{\infty}_0 2\pi \lambda_{p_2} z (1-\exp(-\bar{m}(1-Q_1(\frac{z}{\sigma},\frac{(\frac{P_2}{P_i})^{\frac{1}{\beta}}x}{\sigma}))))dz) \times\\
    \exp(- \pi \lambda_{3} (\frac{P_3}{P_i})^{\frac{2}{\beta}} x^2)dx,
\end{multline}
where 
\begin{align*}
&\mathcal{M}_y(s,x) = \\
& \exp(-2\pi \lambda'_1 x^2 \frac{s \frac{P_1}{P_i} (\frac{a}{x})^{2-\beta}}{\beta(1-\frac{2}{\beta})} {}_{2}F_1[1,1-\frac{2}{\beta};2-\frac{2}{\beta};-s \frac{\frac{P_1}{P_i}}{(\frac{a}{x})^{\beta}}]) \times\\
&\exp(-2\pi \lambda_3' x^2 \frac{s (\frac{P_3}{P_i})^{\frac{2}{\beta}}}{\beta(1-\frac{2}{\beta})} {}_{2}F_1[1,1-\frac{2}{\beta};2-\frac{2}{\beta};-s]) \times\\
&\exp(-2 \pi \lambda_{p_2} \int_{0}^{\infty}  (1-\exp(-\Bar{m} \int_{(\frac{P_2}{P_i})^{\frac{1}{\beta}}x}^{\infty} \frac{1}{1+(s\frac{P_2}{P_i})^{-1}(\frac{y+z}{x})^{\beta}} \\ 
& \frac{1}{\sqrt{2\pi}\sigma^2} \exp(-\frac{y^2}{2\sigma^2})dy))zdz).
\end{align*}

\begin{proof}
See Appendix B.
\end{proof}
\end{theorem}

\subsection{The Average Ergodic Rate in Case 3}
Recall that in case 3, the typical user who is not cache-enabled and the highest power provider is a D2D node but the requested contents are not cached in its local storage; therefore, the only choice for the user is to request the contents to either MBS or SBS (see Fig. 2). In this case, the PDF of the distance between the typical user and the associated serving node (MBS or SBS) is a joint PDF under the event such that a cache-enabled user is the node providing the highest power while either MBS or SBS is the next highest one. Let $x$ be the distance between the typical user and the cache-enabled user providing the highest power, and let $y$ be the distance between the user and the associated serving node in $j$-th tier (MBS or SBS). Given an event $\mathcal{S}_{1,j} = \{P_1 \|\Vec{x}^*_1\|^{-\beta} >  P_j \|\Vec{x}^*_{j}\|^{-\beta} \cap P_j \|\Vec{x}^*_j\|^{-\beta} > P_k \|\Vec{x}^*_k\|^{-\beta}: j, k \in \{2,3\} \; \text{s.t.} \; j\neq k\}$, the joint PDF of the distance between the typical user to the serving node in $j$-th tier is
\begin{multline}
    f_{\mathcal{S}_{1,j}}(x,y) = \frac{\tau_1(x) \tau_j(y)}{\mathbb{P}_{1,j,k}} \exp(-\pi \lambda_1 x^2) \times \\ \exp(-\int^{\infty}_0 2\pi \lambda_{p_2} z(1-\exp(-\bar{m}(1-Q_1(\frac{z}{\sigma},\frac{(\frac{P_2}{P_j})^{\frac{1}{\beta}}y}{\sigma}))))dz) \times\\
    \exp(- \pi \lambda_{3} (\frac{P_3}{P_j})^{\frac{2}{\beta}} y^2),
\end{multline}
where  $j, k \in \{2,3\} \; \text{s.t.} \; j\neq k$ and $\tau_1(x)$ and $\tau_j(y)$ are defined in TABLE 1. The derivation is similar to the Lemma 2 using the result of Proposition 1. Therefore, we skip the derivation. Then, following a similar argument as case 1 and case 2, we obtain the next theorem. 
\begin{theorem}
The average ergodic rate of the typical user connecting to a node in the $j$-th tier in case 3 is
\begin{multline}
    \mathcal{U}_{3,j} = \int^{\infty}_0 \int^{\frac{P_1}{P_j}^{\frac{1}{\beta}}y}_0 \int^{\infty}_0 \frac{\mathcal{M}_y(s,x,y)}{1+s}ds \times \\ 
    \frac{\tau_1(x) \tau_j(y)}{\mathbb{P}_{1,j,k}} \exp(-\pi \lambda_1 x^2) \times \\ \exp(-\int^{\infty}_0 2\pi \lambda_{p_2} z(1-\exp(-\bar{m}(1-Q_1(\frac{z}{\sigma},\frac{(\frac{P_2}{P_j})^{\frac{1}{\beta}}y}{\sigma}))))dz) \times\\
    \exp(- \pi \lambda_{3} (\frac{P_3}{P_j})^{\frac{2}{\beta}} y^2)dxdy,
\end{multline}
where 
\begin{align*}
&\mathcal{M}_y(s,x,y) = \\
& \exp(-2\pi \lambda'_1 y^2 \frac{s \frac{P_1}{P_j} (\frac{x}{y})^{2-\beta}}{\beta(1-\frac{2}{\beta})} {}_{2}F_1[1,1-\frac{2}{\beta};2-\frac{2}{\beta};-s \frac{\frac{P_1}{P_j}}{(\frac{x}{y})^{\beta}}]) \times\\
&\exp(-2\pi \lambda_3' y^2 \frac{s (\frac{P_3}{P_j})^{\frac{2}{\beta}}}{\beta(1-\frac{2}{\beta})} {}_{2}F_1[1,1-\frac{2}{\beta};2-\frac{2}{\beta};-s]) \times\\
&\exp(-2 \pi \lambda_{p_2} \int_{0}^{\infty}  (1-\exp(-\Bar{m} \int_{(\frac{P_2}{P_j})^{\frac{1}{\beta}}y}^{\infty} \frac{1}{1+(s\frac{P_2}{P_j})^{-1}(\frac{y'+z}{y})^{\beta}} \\ 
&\frac{1}{\sqrt{2\pi}\sigma^2} \exp(-\frac{y'^2}{2\sigma^2})dy'))zdz).
\end{align*}
\end{theorem}
\begin{proof}
See Appendix C.
\end{proof}
We denote $\mathcal{U}_4,i$ as the average ergodic rate of the typical user in case 4. The rate is extremely fast since the requested contents can be fetched from their local device immediately. We are interested in the QoS of the downlink data rate among the cases and therefore ignore case 4 for our study.

\section{QoS of Clustered Deployment and Caching Aware Capacity Allocation}

In this section, we analyze the QoS of the cluster deployment and the caching aware capacity allocation system. We define all essential settings and parameters for modeling the traffic flow of caching aware capacity allocation. We follow the same scenarios of the system and the configurations of the parameters as the baseline work (see Section \RomanNumeralCaps{5}, \cite{Yang2016}). \par
\begin{figure*}[!t]
\normalsize
\begin{equation}
\Vec{D}_{8\times4}=
\begin{pmatrix}
\mathcal{G}_{3,1} (1-\alpha) F(1,M_1) & \mathcal{G}_{3,2} (1-\alpha) F(1,M_2) & \mathcal{G}_{3,3} (1-\alpha) & 0 \\
0 & \mathcal{G}_{3,2} (1-\alpha) F(M_2+1,N) & 0 & 0 \\
0 & \mathbb{P}_{2,3} \alpha F(M_1+1,M_2) & \mathbb{P}_{3,2} \alpha F(M_1+1,N) & 0 \\
0 & \mathbb{P}_{2,3} \alpha F(M_2+1,N) & 0 & 0 \\
0 & \mathbb{P}_{1,2,3} (1-\alpha) F(M_1+1,M_2) & \mathbb{P}_{1,3,2} (1-\alpha) F(M_1+1,N) & 0 \\
0 & \mathbb{P}_{1,2,3} (1-\alpha) F(M_2+1,N) & 0 & 0 \\
0 & 0 & 0 & \alpha F(1,M_1) \\
0 & 0 & 0 & 0 \\
\end{pmatrix}.
\end{equation}
\hrulefill
\vspace*{2pt}
\end{figure*}

\subsection{User State and Class}
First, we define the classes of users in all four cases to characterize the presence of users who requests contents under different circumstances. We consider that the typical user is in a certain states depending on which tier to request and whether the request is BH-needed or not. The probability that the typical user is active in different states is expressed in a matrix $\Vec{D}_{8\times4}$, where each entry indicates the probability that the typical user requests contents in in the corresponding states and the columns and the rows of the matrix indicate the state of the active user. This matrix is essentially identical to what is given in the baseline, where they provide it as the comprehensive tabular format (see TABLE \RomanNumeralCaps{1} \cite{Yang2016}). The rows of the matrix indicate all four cases with BH-needed or not, and there are eight different states in total, where we regard those eight states as classes of user. The columns represent tiers from which the user in a certain state requests contents.
For $i=1,...,8$, and $j=1,...,4$, we get the intensity of active users in each single BS/D2D by the thinning property of independent PPP; for example, $\lambda_0 \Vec{D}_{i,j}$ is the intensity of users in the state of $i$-th row and $j$-th column. We define that the set of index of rows represents a class such that $\{\{1,2\},\{3,4\},\{5,6\},\{7,8\}\} \in \{ \text{Class 1}, \text{Class 2}, \text{Class 3}, \text{Class 4} \}$, and the index with even number represents BH-needed while the odd number represents backhaul free (BH-free). The set of index of columns represents a tier where we see that $\{1,2,3,4\} \in \{\text{D2D}, \text{SBS}, \text{MBS}, \text{Local}\}$. Here, "Local" stands for the case when the users who request contents to their own local storage when the contents are cached.

\subsection{User Request Arrival and Service Rate}
Next, we define the total request arrival rate of class of users in each state represented as $(i, j)$ for $i \in \{1,2,...,8\}$ and $j \in \{1,2,3,4\}$. Each MBS cell has the average number of $\frac{\lambda_0}{\lambda_3}$ users, and the rate of a request by each user in that cell is modeled as homogeneous Poisson process with mean arrival rate of $\varsigma$ [requests/s]. Then, the requests of a single user is homogeneous Poisson process with parameter $\frac{\varsigma \lambda_3}{\lambda_0}$. Based on the matrix $\Vec{D}_{8\times4}$, we can obtain the average number of users who are active in the corresponding state in each state. For $i=1,...,8$ and $j=1,...,4$, the average number of users in each state is $\frac{\lambda_0 \Vec{D}_{i,j}}{\lambda_j'}$, where $\lambda_4' = \alpha \lambda_0$. The total request arrival rate in each BS/D2D is ready to be obtained as $\zeta_{i,j}= \frac{\lambda_0 \Vec{D}_{i,j}}{\lambda_j'} \frac{\varsigma \lambda_3}{\lambda_0}$. Each single request by a user consist of a set of contents, and the volume of a set of contents per request is a random variable denoted as $B$ and follows the exponential distribution with mean $\frac{1}{\varrho}$ [contents/request]. Next, we define the average ergodic rate of the user in the state of $(i,j)$ as follows:
\begin{equation}
    \mathbf{A}_{8 \times 4}=
    \begin{cases}
    A_{2m-1,j} = \eta \omega \mathcal{U}_{m,j} \mathbbm{1}\{ \Vec{D}_{2m-1,j} \neq 0\}, \\
    A_{2m,j} = \eta \omega f(\mathcal{U}_{m,j}) \mathbbm{1}\{ \Vec{D}_{2m,j} \neq 0\},
    \end{cases}
\end{equation}
where $m=1,2,3,4$ and $f(\cdot)$ is the effect of backhaul delay, is an arbitrary function. $\eta=1.443$ is the conversion factor between [nats] and [bits], and $\omega$ is bandwidth [Hz] shared among different tiers.
Consider each serving node (MBS, SBS, or D2D) as a server which processes the arrivals of requests from the different class of users, and each user receives the service of downlink transmission at a rate configured by the corresponding cases with BH-needed or BH-free. Then, according to the class of users, a serving node allocates the downlink transmission capacity to communicating users. This system can be viewed as the DPS queue.

\subsection{DPS Queue and QoS Metric}
Let $\mathbb{D}:=\{1,2,...,8\}$ be the set of class, and let $ \{X_j(t):t\geq0\}$ be the process of the number of users who request contents to a tier for $j \in \{1,2,3,4\}$ with a vector $\Vec{x}_j = (x_{1,j},x_{2,j},...,x_{8,j})$ which is counting the number of requests in each class. We introduce weights $w_{1,j}, w_{2,j},...,w_{8,j}$ to the class of users to differentiate the priority of processing their requests. Then, we claim that $ \{X_j(t):t\geq0\}$ has discrete state space $\mathbb{N}^{\mathbb{D}}$ with a continuous-time Markov process generated by
\begin{equation}
    \begin{cases}
    q(\Vec{x}_j,\Vec{x}_j+\Vec{\epsilon}_i) = \zeta_{i,j}, & \Vec{x}_j \in \mathbb{N}^{\mathbb{D}}, \\
    q(\Vec{x}_j,\Vec{x}_j-\Vec{\epsilon}_i) = \frac{\Vec{A}_{i,j}}{ S/\varrho}\frac{x_{i,j}}{x_{\mathbb{D},j}} \frac{w_{i,j}}{w_{\mathbb{D},j}}, & \Vec{x}_j \in \mathbb{N}^{\mathbb{D}}, \Vec{x}_j > 0,
    \end{cases}
\end{equation}
where $\Vec{\epsilon}_i$ is the vector of $\mathbb{N}^\mathbb{D}$ with 1 in $i$-th element ($i=1,...,8$) and 0 elsewhere. $x_{\mathbb{D},j} = \sum_{i \in \mathbb{D}} x_{i,j}$ is the total number of requests in the queue, and $w_{\mathbb{D},j}=\sum_{i\in \mathbb{D}} w_{i,j}$ is the total weights of all classes. The class traffic demand can be obtained as $\rho_{i,j} = \frac{\zeta_{i,j}S}{\varrho}$, and the critical traffic value which the queue will be steady state is $\rho_{c,j} = \frac{\rho_{\mathbb{D},j}}{\sum_{i\in \mathbb{D}}\rho_{i,j}\Vec{A}^{-1}_{i,j}}$ where $\rho_{\mathbb{D},j}= \sum_{i \in \mathbb{D}}\rho_{i,j}$ (see \cite{Mohamed2013, Yang2016}). We define a tier traffic intensity of DPS queue as $\rho^{'}_{i,j} = \frac{\zeta_{i,j}}{S/(\varrho \Vec{A}_{i,j})} \delequal \frac{\lambda_{(i,j)}}{\mu_{(i,j)}}$, where $S/(\varrho \Vec{A}_{i,j})$ is the rate of completion of transmission of requests [requests/s]. Recall that $S$ is the size of each content. We need to obtain mean sojourn time of DPS queue described thus far; however, the total class of users is large, and this makes deriving the exact mean sojourn time intractable. To this end, we use the best approximated mean sojourn time of DPS queue (see \cite{izagirre2017}). We assume that the traffic load of the tier is less than the critical traffic value, which is less than 1, and all weights are strictly positive, which implies the queue is stable. We also assume that arrivals of requests do not change over time and finite. This yields the condition that the queue is stable and ergodic; Little's law is applicable.
\begin{proposition}
For $\sum_{i\in \mathbb{D}}\rho^{'}_{i,j} < \rho_{c,j} < 1$, all $w_{i,j}>0$, and $\mu_{(i,j)}<\infty$, 
approximated mean number of requests, delay, and throughput per user request of the $i$-th class at $j$-th tier for $i \in \mathbb{D}$ and $j \in \{1,2,3,4\}$ can be obtained as,
\begin{align}
    \Bar{N}^{INT}_{i,j} &= \lambda_{(i,j)}\Bar{S}^{INT}_{(i,j)}(\lambda_{(i,j)},\mu_{(i,j)},w_{1,j},w_{2,j},...,w_{8,j}),\\
    \Bar{D}^{INT}_{i,j} &=\Bar{S}^{INT}_{(i,j)}(\lambda_{(i,j)},\mu_{(i,j)},w_{1,j},w_{2,j},...,w_{8,j}),\\
    \Bar{T}^{INT}_{i,j} &=\frac{\rho^{'}_{i,j}}{\Bar{N}^{INT}_{i,j}},
\end{align} 
where the approximated sojourn time of DPS queue is
\begin{multline}
    \Bar{S}^{INT}_{(i,j)}(\lambda_{(i,j)},\mu_{(i,j)},w_{1,j},w_{2,j},...,w_{8,j}) = \\ \frac{1}{\mu_{(i,j)}} + \frac{1}{\mu_{(i,j)}}\sum_{k=1}^{8}(\frac{\lambda_{(k,j)}}{\mu_{(k,j)}}) + \\ \sum_{k=1}^{8}(\frac{\lambda_{(k,j)}}{\mu_{(k,j)}}\frac{w_{k,j}-w_{i,j}}{w_{k,j}\mu_{(k,j)}-w_{i,j}\mu_{(i,j)}}) + \\
    \frac{(\sum_{k=1}^{8}(\frac{\lambda_{(k,j)}}{\mu_{(k,j)}}))^2}{1-\sum_{k=1}^{8}(\frac{\lambda_{(k,j)}}{\mu_{(k,j)}})}\frac{1}{w_{i,j}\mu_{(i,j)}}\frac{\sum_{k=1}^{8}(\frac{\lambda_{(k,j)}/\lambda_{(\mathbb{D},j)}}{\mu_{(k,j)}^2})}{\sum_{k=1}^{8}(\frac{\lambda_{(k,j)}/\lambda_{(\mathbb{D},j)}}{\mu_{(k,j)}^2w_{k,j}})}.
\end{multline}
\begin{proof}
Since $\sum_{i\in \mathbb{D}}\rho^{'}_{i,j} < \rho_{c,j} < 1$ and the weight of class $i$ is strictly positive and the service time distribution is finite, then this is sufficient that DPS queue is stable (see Theorem 1, \cite{Avrachenkov2005}). We assumed that the arrival of the requests is ergodic. The approximated mean unconditional sojourn time when service time is exponentially distributed is given by using light and heavy traffic approximation order 2 (see \cite{izagirre2017}). The equation (28) follows from Little's law, and (29) and (30) are by the definition (see Proposition 1 \cite{Mohamed2013}).
\end{proof}
\end{proposition}
This study aims to see the effectiveness and behavior of caching aware capacity allocation system, where its simulation is infeasible. The accuracy of the approximation of mean unconditional sojourn time of DPS queue has already studied extensively by \cite{izagirre2017}. Therefore, we explore the results of numerical calculation only.

\section{Numerical Results}

In this section, we show the numerical results of the clustered deployment of SBSs and caching aware capacity allocation system. Also, the numerical results of the non-allocated system are presented for elucidating the effect of clustered deployment of SBSs compering with the case of non-clustered deployment of SBSs which is our baseline examined by the previous work done by \cite{Yang2016}.

\begin{figure}[!htbp]
\includegraphics[width=8.8cm,height=5.5cm]{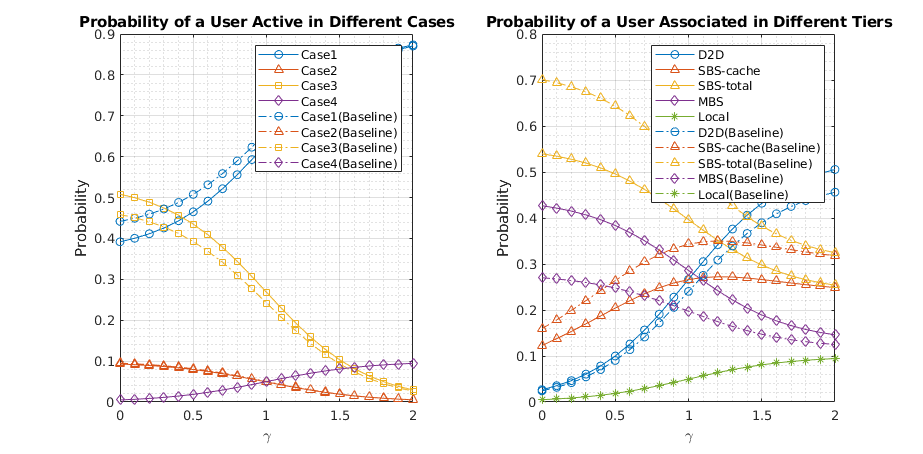}
\centering
\caption{The probability of a user active in different cases and associated in different tiers; $\alpha=0.1$, $\{P_1,P_2,P_3\} = \{3,13,193\}$, $\beta=4$, $\bar{m}=10$, $\{\lambda_0, \lambda_1, \lambda_2, \lambda_3\} = \{ \frac{1000}{\pi*1000^2}, \lambda_0*\alpha, \frac{3*\Bar{m}}{\pi*1000^2}, \frac{2}{\pi*1000^2}\}$, and $ \text{  when the case of baseline, } \lambda_2 = \frac{30}{\pi*1000^2}$. The variance of TCP is $\sigma=250$.}
\end{figure}

Fig.3 shows the probability of a user active in different cases (left) and different tiers (right) with changing the parameter of the distribution of contents. Recall that $\gamma$ is the parameter of content distribution. The larger value of $\gamma$ implies that the distribution is steeper and dominated by popular contents. From the plot on the left, we see that only case 1 and case 3 are sensitive to the distribution of contents. When the popularity of the content is even, the clustered deployment of SBSs reduces the probability of user active in case 1 and increase the probability in case 3. The plot on the right coincides with our intuition; the clustered deployment of SBSs reduces the probability of association to itself from users whereas it increases the probability of association to MBSs. Also, as $\gamma$ increases, the probability of user associated with D2D transmission increases, and when the deployment of SBSs are clustered, the increase is accelerated.

Fig. 4 compares the average ergodic rate of each case (left) and the total average ergodic rate (right) of both the baseline of \cite{Yang2016} and clustered deployment of SBSs. The total average ergodic rate is the weighted sum by the probability of the user active in each case. From the plot on the right, we can confirm that when the deployment of SBSs is clustered, more cache-enabled users are needed to be active to support the requests of the contents from users and the active D2D links yield more interference to degrade the total average data rate than that of the baseline.

\begin{figure}
    \centering
    \includegraphics[width=0.225\textwidth,height=0.225\textheight]{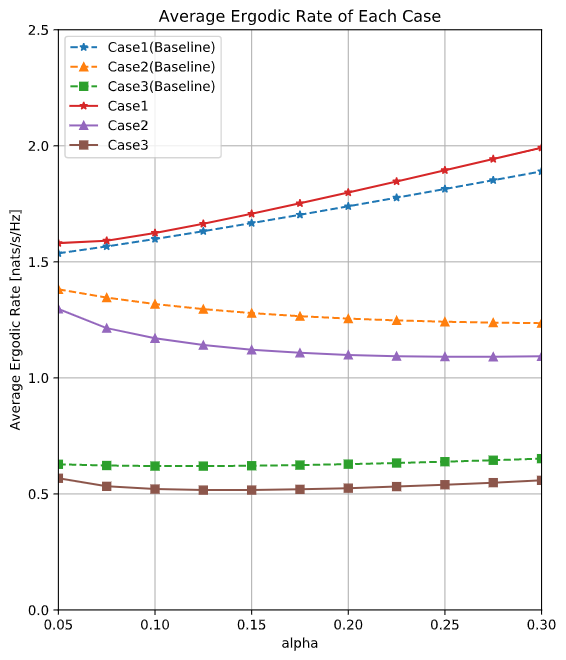}
    \includegraphics[width=0.225\textwidth,height=0.225\textheight]{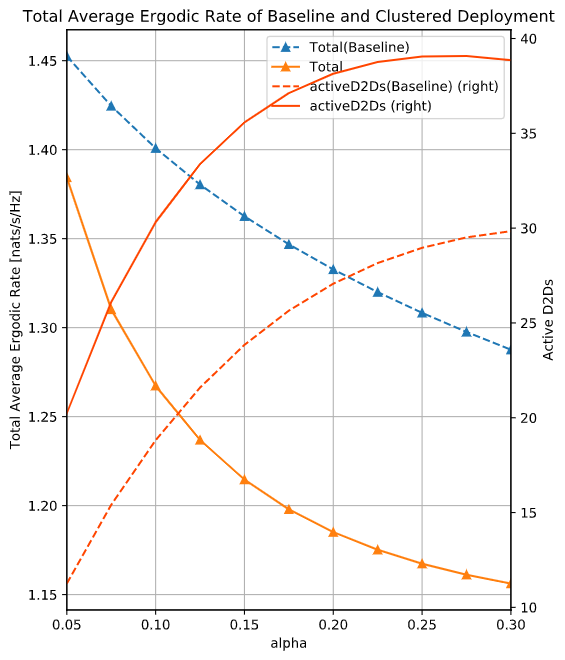}
    \caption{The average ergodic rate in different cases with varying $\alpha$; $\gamma = 0.8$, $\{P_1,P_2,P_3\} = \{73,373,1773\}$, $\beta=4$, $\bar{m}=10$, $A=1000^2$, $\{\lambda_0, \lambda_1, \lambda_2, \lambda_3\} = \{ \frac{300*A}{\pi*500^2}, \lambda_0*\alpha, \frac{3*\Bar{m}*A}{\pi*500^2}, \frac{6*A}{\pi*500^2}\}$, and $ \text{  when baseline case is } \lambda_2 = \frac{30*A}{\pi*500^2}$. The variance of TPP $\sigma=0.05$.}
    \label{fig:foobar}
\end{figure}

Fig. 5 shows the mean number of requests and mean throughput in each SBS under the case of baseline or clustered deployment. When SBSs are clustered, the mean number of requests is less than the case when SBSs are scattered uniformly. This is because each SBS is proximal to the other SBSs and the requests from users who are associated to SBSs are decentralized by those SBSs in a cluster, and also, the reason is the smaller coverage area of SBSs compared to the baseline case. This result confirms the result from Fig. 3 where the probability of a user associated with SBS is smaller than that of the baseline. Next, while $\alpha$ (the ratio of cache-enabled users) increases, the amount of case 1 users belonging to class 1 and 2 decreases significantly; however, the amount of class 2 and 3 users which are class 3 to class 6 increases. The plot also shows the result of caching aware capacity allocation system. We found that assigning more weights on class 5 and 6 (users in case 3) improve the throughput and traffic load not only those users but the entire networks. This improvement is conspicuous as the number of cache-enabled users increases. 

\begin{figure}
    \centering
    \includegraphics[width=0.225\textwidth,height=0.225\textheight]{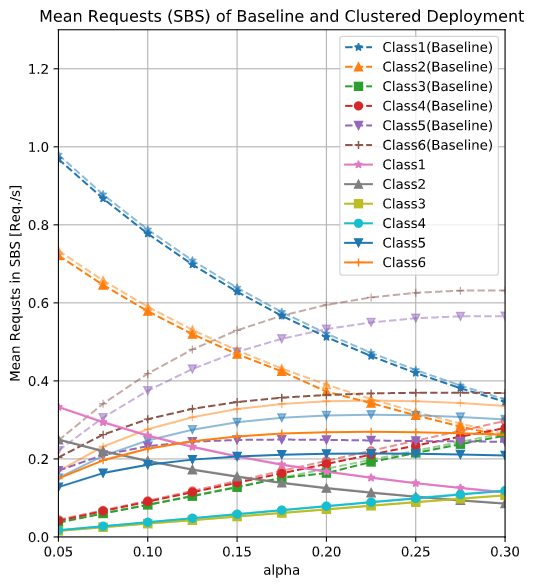}
    \includegraphics[width=0.225\textwidth,height=0.225\textheight]{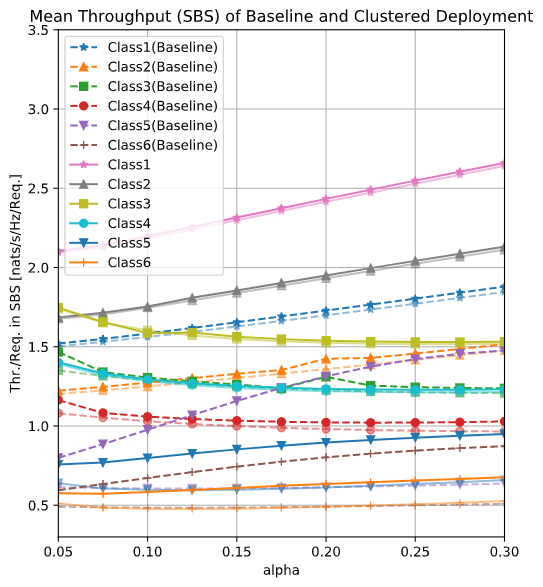}
    \caption{The mean number of requests and the mean throughput of SBS; the lines with pale color are PS queue results, and the lines with bright color are DPS queue results. The weights for baseline and clustered deployment:$\{w_{1,2}, w_{2,2}, w_{3,2},w_{4,2},w_{5,2},w_{6,2}\} = \{1,1,1.1,1.1,1.5,1.87\}$ with $\varsigma = 0.2$, $S = 100$[Mbits], $\omega = 70$MHz, $\varrho = 1$; parameters are the same as Fig. 4.}
    \label{fig:foobar}
\end{figure}

\begin{figure}
    \centering
    \includegraphics[width=0.225\textwidth,height=0.225\textheight]{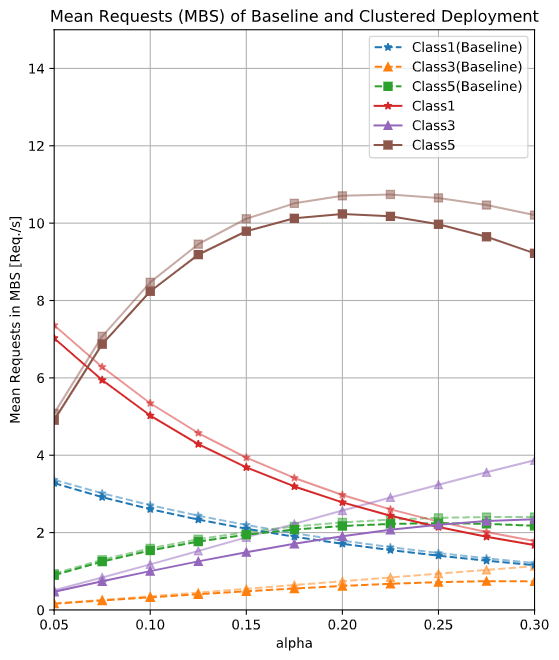}
    \includegraphics[width=0.225\textwidth,height=0.225\textheight]{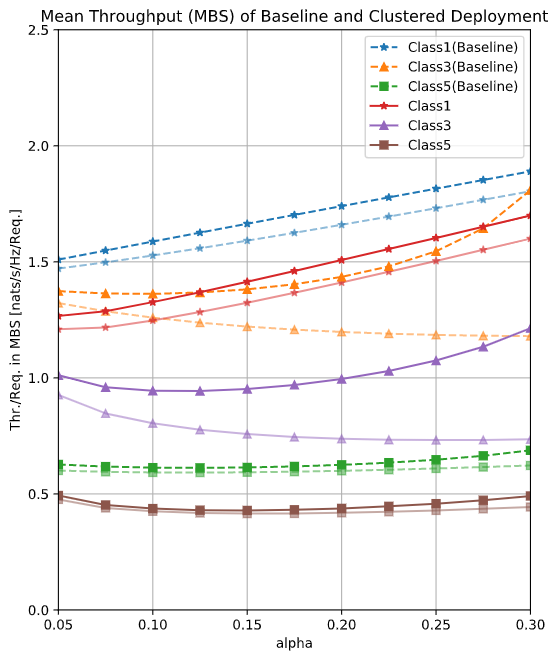}
    \caption{The mean number of requests and the mean throughput of MBS; the lines with pale color are PS queue results, and the lines with bright color are DPS queue results. The weights for baseline:$\{ w_{1,3},w_{3,3},w_{5,3}\} = \{1,1,1.8\}$ and for clustered deployment $\{ w_{1,3},w_{3,3},w_{5,3}\} = \{1,1,1.5\}$ with $\varsigma = 0.2$, $S = 100$[Mbits], $\omega = 70$MHz, $\varrho = 1$; parameters are the same as Fig. 4.}
    \label{fig:foobar}
\end{figure}

Fig. 6 compares the result of caching aware capacity allocated and non-allocated system in MBSs. First, we found that by weighting more on processing the requests from the users in class 5, the throughput can be increased, and the traffic load on MBSs can be alleviated significantly. This is similar to the result of Fig. 5. Next, we found that the traffic load of MBSs increases under the clustered deployment of SBSs and this increase of traffic load is conspicuous from class 5 users. Also, this result confirms the result from Fig. 3 where the probability of a user associated with MBSs is larger than that of baseline.

\begin{figure}
    \centering
    \includegraphics[width=0.225\textwidth,height=0.225\textheight]{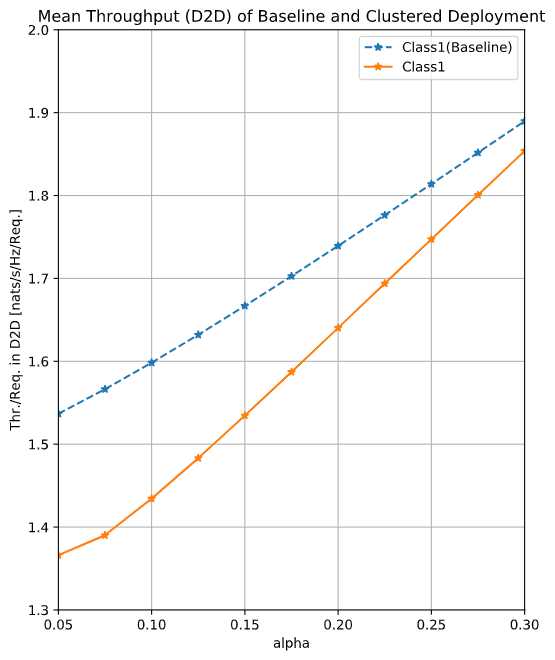}
    \caption{The mean throughput of D2D based on PS queue with $\varsigma = 0.2$, $S = 100$[Mbits], $\omega = 70$MHz, $\varrho = 1$; parameters are the same as Fig. 4.}
    \label{fig:foobar}
\end{figure}

Fig. 7 shows the mean throughput of D2D transmission which has only class 1. The plot confirms that the throughput of D2D transmission is less efficient under clustered deployment of SBSs compared to the baseline case.

\begin{figure*}[!t]
\normalsize
\begin{align}
&\mathcal{M}_y(s,x) = \mathbb{E}_{h_{\Vec{y}},\Phi_j'}[ \exp(-s(\sum^{3}_{j=1} \sum_{\Vec{y} \in \Phi_j' \setminus \{\Vec{x}^*_i\}} \frac{P_j}{P_i} h_{\Vec{y}} \left(\frac{\|\Vec{y}\|}{x} \right)^{-\beta}+ \frac{N_0}{P_i x^{-\beta}}))] \nonumber \\
&\overset{(a)}{=}\prod_{j\in \{1,3\}} \mathbb{E}_{\Phi_j'}[ \prod_{\Vec{y} \in \Phi_j' \setminus \{\Vec{x}^*_i\}} \frac{1}{1+s\frac{P_j}{P_i} (\frac{\|\Vec{y}\|}{x})^{-\beta}}] \mathbb{E}_{\Phi_{p_2},\Psi_i}[\prod_{\Vec{z} \in \Phi_{p_2}} \prod_{\Vec{y} \in \Psi_{[\Vec{z}]} \setminus \{\Vec{x}^*_i\}} \frac{1}{1+s\frac{P_2}{P_i} (\frac{\|\Vec{y}\|}{x})^{-\beta}}] \exp(-s\frac{N_0}{P_i x^{-\beta}}) \nonumber \\
&\overset{(b)}{=} \prod_{j\in \{1,3\}} \exp(- \lambda'_j \int_{\mathbb{R}^2 \setminus b_0((\frac{P_j}{P_i})^{\frac{1}{\beta}}x)}\frac{1}{1+(s\frac{P_j}{P_i})^{-1}(\frac{\|\Vec{y}\|}{x})^{\beta}}d\Vec{y})\exp(-\lambda_{p_2} \int_{\mathbb{R}^2}  (1-\exp(-\Bar{m} \int_{\mathbb{R}^2 \setminus b_0((\frac{P_2}{P_i})^{\frac{1}{\beta}}x)} \nonumber \\ &\frac{\frac{1}{\sqrt{2\pi}\sigma^2} \exp(-\frac{y^2}{2\sigma^2})}{1+(s\frac{P_2}{P_i})^{-1}(\frac{\|\Vec{y}+\Vec{z}\|}{x})^{\beta}} d\Vec{y}))d\Vec{z}) \exp(-s\frac{N_0}{P_i x^{-\beta}}) \nonumber \\
&\overset{(c)}{=}\prod_{j\in \{1,3\}} \exp(- 2 \pi \lambda'_j \int_{(\frac{P_j}{P_i})^{\frac{1}{\beta}}x}^{\infty}\frac{1}{1+(s\frac{P_j}{P_i})^{-1} (\frac{y}{x})^{\beta}}ydy)\exp(-2 \pi \lambda_{p_2} \int_{0}^{\infty}  (1-\exp(-\Bar{m} \int_{(\frac{P_2}{P_i})^{\frac{1}{\beta}}x}^{\infty} \nonumber \\
&\frac{ \frac{1}{\sqrt{2\pi}\sigma^2}\exp(-\frac{y^2}{2\sigma^2})}{1+(s\frac{P_2}{P_i})^{-1}(\frac{y+z}{x})^{\beta}}dy)) z dz) \exp(-s\frac{N_0}{P_ix^{-\beta}}) \nonumber \\
&\overset{(d)}{=} \prod_{j\in \{1,3\}} \exp(-2\pi \lambda'_j x^2 \frac{s (\frac{P_j}{P_i})^{\frac{2}{\beta}}}{\beta(1-\frac{2}{\beta})} {}_{2}F_1[1,1-\frac{2}{\beta};2-\frac{2}{\beta};-s]) \exp(-2 \pi \lambda_{p_2} \int_{0}^{\infty}  \exp(-\Bar{m} \int_{(\frac{P_2}{P_i})^{\frac{1}{\beta}}x}^{\infty} \nonumber \\
&\frac{ \frac{1}{\sqrt{2\pi}\sigma^2}\exp(-\frac{y^2}{2\sigma^2})}{1+(s\frac{P_2}{P_i})^{-1}(\frac{y+z}{x})^{\beta}}dy)) z dz) \exp(-s\frac{N_0}{P_ix^{-\beta}}).
\end{align}
\hrulefill
\vspace*{2pt}
\end{figure*}

\section{Discussion and Conclusion}
In this work, we have elucidated the performance of three tier cache-enabled HetNets consisting of macro base stations (MBSs), small base stations (SBSs), and device-to-device sharing links (D2Ds) under the clustered deployment of SBSs and proposed the capacity allocation system according to the caching circumstances. We found that when the deployment of SBSs is clustered the performance can be differed significantly compared to the case in which SBSs are uniformly scattered. This is because each SBS is proximal to the other SBSs, and the requests from users associated with SBSs are decentralized by those SBSs in a cluster, and also, the reason is the smaller coverage area of SBSs compared to the baseline case. The increase of traffic in MBSs is conspicuous from the non-cache-enabled users who are struggling from the interference from active D2D links. We also found that a larger amount of cache-enabled users are needed to be active as D2D transmitters when SBSs are clustered. As a result, the interference from active D2D links deteriorates the efficiency of the overall data rate of the networks, and this yields more consumption of energy in user devices. We conclude that although employing caching in D2D sharing can improve the throughput of HetNets, the careful management of interference of those D2D links are essential to utilize available resources efficiently. On the other hand, we also found that capacity allocation according to the caching circumstances can improve the throughput of the entire network and alleviate the traffic load. Such an allocation system can be modeled as an elegant queuing theory of DPS queue, and the throughput and the other QoS metrics of the system are derived by using the approximated mean sojourn time. Allocating more resources to serve the requests from a class of users that are interfered with active D2D links can improve not only the throughput of that class of users but also the entire network. As the number of cache-enabled users increases, the throughput gain under the allocation system is significant. Our findings might be practically important when it comes to carrying out D2D caching in HetNets. These findings suggest that smart management of interference from D2D transmissions is essential, and it poses further research direction to incorporating scheduling of D2D transmissions with caching. Considering the determinantal scheduling proposed by \cite{blaszczy2020} into the system might be an interesting point to start.

\section{Appendix}

\subsection{proof of Theorem 1}
Let $x$ be the distance between the typical user and the serving node under the max-power association law denoted as $\|\Vec{x}^*_i\|$. We see that from (15)
\begin{align}
\mathcal{U}_{1,i} &= \int^{\infty}_0\mathbb{E}_{SINR_{i}}[\ln(1+SINR_{i}(x))|\|\Vec{x}^*_i\|=x]f_{\mathcal{S}_i}(x)dx \nonumber \\
&= \int^{\infty}_0 \int^{\infty}_0 \frac{\mathcal{M}_y(s,x)-\mathcal{M}_{xy}(s,x)}{s}ds f_{\mathcal{S}_i}(x)dx,
\end{align}
where 
\begin{multline*}
    \mathcal{M}_y(s,x) = \mathbb{E}_{h_{\Vec{y}},\Phi_j'}[ \exp(-s(\sum^{3}_{j=1} \sum_{\Vec{y} \in \Phi_j' \setminus \{\Vec{x}^*_i\}} \frac{P_j}{P_i} h_{\Vec{y}} \\
    \left(\frac{\|\Vec{y}\|}{x} \right)^{-\beta}+\frac{N_0}{P_i x^{-\beta}}))]
\end{multline*}
and
\begin{multline*}
    \mathcal{M}_{xy}(s,x) = \mathbb{E}_{h_{\Vec{x}^*_i},h_{\Vec{y}},\Phi_j'}[ \exp(-s( h_{\Vec{x}^*_i} +  \\
    \sum^{3}_{j=1} \sum_{\Vec{y} \in \Phi_j' \setminus \{\Vec{x}^*_i\}} \frac{P_j}{P_i} h_{\Vec{y}} \left(\frac{\|\Vec{y}\|}{x}\right)^{-\beta} + \frac{N_0}{P_ix^{-\beta}}))].
\end{multline*}
The second equality uses the lemma 1 given by Hamdi in 2010 (see, \cite{Hamdi2010}). Let $\Psi_{[\Vec{z}]}$ be the daughter points process around $\Vec{z} \in \Phi_{p_2}$. Then, $\mathcal{M}_y(s)$ can be derived as (32), where we move $P_i$ and $\|\Vec{x}_i^*\|$ to the denominator of SINR. The equation in $(a)$ is valid since we apply the Laplace transform of $h_{\Vec{y}}$, and all three tiers are independent. For $(b)$, all interfering nodes are far from the serving node, and we use PGFL of PPP and PPCP \cite{Saha2017_2,Haenggi2012}. $(c)$ holds by expressing Cartesian space to polar coordinates. $(d)$ holds by using change of variable $u=(\frac{y}{x})^{\beta}$ and integral expression from the table of integral (see 3.194 page 315, \cite{Gradshtein1996}). Then, for the case of $\mathcal{M}_{x y}(s)$,
\begin{align}
   \mathcal{M}_{x y}(s, x) &= \mathbb{E}_{h_{\Vec{x}^*_i},h_{\Vec{y}},\Phi_j'}[ \exp(-s( h_{\Vec{x}^*_i} + \nonumber \\
   & \sum^{3}_{j=1} \sum_{\Vec{y} \in \Phi_j \setminus \{\Vec{x}^*_i\}} \frac{P_j}{P_i} h_{\Vec{y}} \left(\frac{\|\Vec{y}\|}{x}\right)^{-\beta} + \frac{N_0}{P_i x^{-\beta}}))] \nonumber \\
   &\overset{(a)}{=} \mathcal{L}_{h_{\Vec{x}^*_i}} (s) \mathcal{M}_y(s, x) \nonumber \\
   &= \frac{1}{1+s}\mathcal{M}_y(s, x),
\end{align}
where (a) holds by independence of random variables and Laplace transform of $h_{\Vec{x}_i}$. plug in (32) to (35), we get our result. Note that the background noise is dominated by the interference in HetNets, and therefore, we let $N_0 \rightarrow 0$.

\subsection{proof of Theorem 2}
The derivation of Theorem 2 is similar to Theorem 1, except the PGFL of interfering nodes from tier 1 which is the Poisson point process of active D2D transmitter. Let denote $a$ be the distance to the closest possible active D2D transmitter, and since $a$ can be small as 0, we derive the PGFL of interference from tier 1 as
\begin{align*}
    &\mathbb{E}_{\Phi_1'}[ \prod_{\Vec{y} \in \Phi_1' \setminus \{\Vec{x}^*_i\}} \frac{1}{1+s\frac{P_j}{P_i} (\frac{\|\Vec{y}\|}{x})^{-\beta}}] =
    \exp(- \lambda'_1 \int_{\mathbb{R}^2 \setminus b_0(a)} \\
    &\frac{1}{1+(s\frac{P_1}{P_i})^{-1}(\frac{\|\Vec{y}\|}{x})^{\beta}}d\Vec{y})\\
    &=\exp(- 2 \pi \lambda'_1 \int_{a}^{\infty}\frac{1}{1+(s\frac{P_1}{P_i})^{-1}(\frac{y}{x})^{\beta}}ydy)\\
    &\overset{(a)}{=}\exp(-2\pi \lambda'_1 x^2 \frac{s \frac{P_1}{P_i} (\frac{a}{x})^{2-\beta}}{\beta(1-\frac{2}{\beta})} {}_{2}F_1[1,1-\frac{2}{\beta};2-\frac{2}{\beta};-s \frac{\frac{P_1}{P_i}}{(\frac{a}{x})^{\beta}}]).
\end{align*}
For the equality (a) we uses the change of variable $u=(\frac{y}{x})^{\beta}$ and the integral expression from the table of integral (see 3.194 page 315, \cite{Gradshtein1996}). Replacing this with the PGFL of interference from tier 1 derived in Theorem 1, we get our result.

\subsection{proof of Theorem 3}
Let $x$ be the distance between the typical user and the closest active D2D node and $y$ be the distance between the typical user and the closest serving node in $j$-th tier for $j \in \{2,3\}$, and let $y'$ be the distance of interfering nodes. By following the same argument as the proof of the Theorem 1 (Appendix A.), we only need to obtain $\mathcal{M}_y(s,x,y)$, where we see that
\begin{align}
\mathcal{U}_{3,j} &= \int^{\infty}_0\mathbb{E}_{SINR_{i}}[\ln(1+SINR_{i}(x))|x,y]f_{\mathcal{S}_{1,j}}(x,y)dxdy \nonumber \\
&= \int^{\infty}_0 \int^{\infty}_0 \frac{\mathcal{M}_y(s,x,y)-\mathcal{M}_{xy}(s,x,y)}{s}ds f_{\mathcal{S}_{1,j}}(x,y)dxdy \nonumber \\
&= \int^{\infty}_0 \int^{\infty}_0 \frac{\mathcal{M}_y(s,x,y)}{1+s}ds f_{\mathcal{S}_{1,j}}(x,y)dxdy.
\end{align}
Then, we derive $\mathcal{M}_y(s,x,y)$ as (33). The equality of (a) holds by using PGFL of PPP and PPCP \cite{Saha2017_2, Haenggi2012}. The distance of the closest active D2D transmitters can be any range, and the distance of interfering nodes of MBSs and SBSs are far from the serving node in tier $j \in \{2,3\}$. Therefore, the distance of active D2D nodes in the case 3 is in the range of $0<x<(\frac{P_1}{P_j})^{\frac{1}{\beta}}y$. The rest of the proof is similar to Theorem 1.

\begin{figure*}[!t]
\normalsize
\begin{align}
&\mathcal{M}_y(s,x,y) = \mathbb{E}_{h_{\Vec{y}},\Phi_n'}[ \exp(-s(\sum^{3}_{n=1} \sum_{\Vec{y'} \in \Phi_n \setminus \{\Vec{x}^*_j\}} \frac{P_n}{P_j} h_{\Vec{y'}} \left(\frac{\|\Vec{y'}\|}{y}\right)^{-\beta}+ \frac{N_0}{P_j y^{-\beta}}))] \nonumber \\
&=\prod_{n\in \{1,3\}} \mathbb{E}_{\Phi_n'}[ \prod_{\Vec{y'} \in \Phi_n' \setminus \{\Vec{x}^*_j\}} \frac{1}{1+s\frac{P_n}{P_j} (\frac{\|\Vec{y'}\|}{y})^{-\beta}}] \mathbb{E}_{\Phi_{p_2},\Psi_i}[\prod_{\Vec{z} \in \Phi_{p_2}} \prod_{\Vec{y'} \in \Psi_{[\Vec{z}]} \setminus \{\Vec{x}^*_j\}} \frac{1}{1+s\frac{P_2}{P_j} (\frac{\|\Vec{y'}\|}{y})^{-\beta}}] \exp(-s\frac{N_0}{P_j y^{-\beta}}) \nonumber \\
&\overset{(a)}{=} \exp(- \lambda'_1 \int_{\mathbb{R}^2 \setminus b_0(x)}\frac{1}{1+(s\frac{P_1}{P_j})^{-1}(\frac{\|\Vec{y'}\|}{y})^{\beta}}d\Vec{y'}) \exp(- \lambda_3 \int_{\mathbb{R}^2 \setminus b_0((\frac{P_3}{P_j})^{\frac{1}{\beta}}y)}\frac{1}{1+(s\frac{P_3}{P_j})^{-1}(\frac{\|\Vec{y'}\|}{y})^{\beta}}d\Vec{y'}) \times \nonumber \\
&\exp(-\lambda_{p_2} \int_{\mathbb{R}^2}  (1-\exp(-\Bar{m} \int_{\mathbb{R}^2 \setminus b_0((\frac{P_2}{P_j})^{\frac{1}{\beta}}y)} \frac{\frac{1}{\sqrt{2\pi}\sigma^2} \exp(-\frac{y^2}{2\sigma^2})}{1+(s\frac{P_2}{P_j})^{-1}(\frac{\|\Vec{y'}+\Vec{z}\|}{y})^{\beta}} d\Vec{y'}))d\Vec{z}) \exp(-s\frac{N_0}{P_j y^{-\beta}})  \nonumber \\
&\overset{(b)}{=} \exp(-2\pi \lambda'_1 y^2 \frac{s \frac{P_1}{P_j} (\frac{x}{y})^{2-\beta}}{\beta(1-\frac{2}{\beta})} {}_{2}F_1[1,1-\frac{2}{\beta};2-\frac{2}{\beta};-s \frac{\frac{P_1}{P_j}}{(\frac{x}{y})^{\beta}}]) \exp(-2\pi \lambda_3 y^2 \frac{s (\frac{P_3}{P_j})^{\frac{2}{\beta}}}{\beta(1-\frac{2}{\beta})} {}_{2}F_1[1,1-\frac{2}{\beta};2-\frac{2}{\beta};-s]) \times \nonumber \\
&\exp(-2 \pi \lambda_{p_2} \int_{0}^{\infty}  (1-\exp(-\Bar{m} \int_{(\frac{P_2}{P_j})^{\frac{1}{\beta}}y}^{\infty} \frac{\frac{1}{\sqrt{2\pi}\sigma^2} \exp(-\frac{y'^2}{2\sigma^2})}{1+(s\frac{P_2}{P_j})^{-1}(\frac{y'+z}{y})^{\beta}}dy'))zdz)\exp(-s\frac{N_0}{P_i y^{-\beta}}).
\end{align}
\hrulefill
\vspace*{2pt}
\end{figure*}

\printbibliography
\end{document}